\documentclass[a4paper]{article}
\usepackage[a4paper, margin=2.5cm]{geometry}
\usepackage{microtype}
\usepackage{graphicx,amssymb,amsmath}
\usepackage{color}
\usepackage{lipsum}
\usepackage{array}
\usepackage{wrapfig}

\usepackage{hyperref}
\usepackage{gensymb}
\usepackage[capitalise,noabbrev]{cleveref}
\usepackage{thm-restate}

\newtheorem{theorem}{Theorem}

\newtheorem{lemma}[theorem]{Lemma}

\def\QED{\ensuremath{{\square}}}
\def\markatright#1{\leavevmode\unskip\nobreak\quad\hspace*{\fill}{#1}}
\newenvironment{proof}
  {\begin{trivlist}\item[\hskip\labelsep{\bf Proof.}]}
  {\markatright{\QED}\end{trivlist}}

\declaretheorem[name=Claim,sibling=theorem]{claim}
\declaretheorem[name=Corollary,sibling=theorem]{corollary}
\declaretheorem[name=Observation,sibling=theorem]{observation}

\newcommand{\edgedistance}{\ensuremath{h}}   %
\newcommand{\cornerdistance}{\ensuremath{d}}   %
\newcommand{\tangentcircle}{\ensuremath{R}}   %
\newcommand{\polar}{\ensuremath{\circ}}   %
\newcommand{\calI}{\ensuremath{\mathcal{I}}}   %

\title{On Reconstructing a Convex Polygon from Partial Information\footnotemark[1]}

 \author{Alexander Baumann\footnotemark[2] \and Therese Biedl\footnotemark[3]  \and  Mahmoud Elashmawi\footnotemark[2] \and Simon D.~Fink\footnotemark[4] \and  Maria Saumell\footnotemark[5]  \and  Andr\'{e} Schulz\footnotemark[6]}

\begin{document}

\maketitle

\renewcommand{\thefootnote}{\fnsymbol{footnote}}
\footnotetext[1]{This work was initiated at the European Research Week on Geometric Graphs 2025, held in Chorin (Germany). We thank the organizers for a fruitful atmosphere.   
We also thank Wolfgang Mulzer for helpful input.}
\footnotetext[2]{Institut für Informatik, Freie Universität Berlin, Germany.}
\footnotetext[3]{David R.~Cheriton School of Computer Science, University of Waterloo, Canada.   Research done while visiting Freie Universität Berlin, Germany.   Supported by NSERC.}
\footnotetext[4]{Algorithms and Complexity Group, TU Wien, Austria.}
\footnotetext[5]{Department of Theoretical Computer Science, Faculty of Information Technology, Czech Technical University in Prague, Czech Republic. Supported by the Czech Science Foundation, grant number 23-04949X.}
\footnotetext[6]{FernUniversität in Hagen, Germany.}
\renewcommand{\thefootnote}{\arabic{footnote}}

\begin{abstract}
The reconstruction problem asks to construct a (convex) polygon that has a specified
set of features, such as an ordered set of edge-lengths or an ordered set of 
polygon-angles.   In this paper, we do a systematic exploration 
of the reconstruction problem in
all scenarios where one or two sets of 
features have been specified.    Some of these scenarios were well-studied
already, for some we develop testing algorithms and/or hardness results, and 
many give rise to interesting open problems for future study.
\end{abstract}

\section{Introduction}

Consider a %
polygon $P$ with $n\geq 3$ corners.
This polygon has a number of \emph{features}, for example the lengths
of the $n$ edges, or the $n$ angles at the corners, which can be computed
from the coordinates of the corners.  In this work, we study
the reverse \emph{reconstruction question}:   
Given some features of an unknown polygon $P$,
is there a \emph{realization}, i.e., a polygon $P$ that has these features?
Can we efficiently test whether there is one?
If so, can we
construct $P$ efficiently and is it unique (at least up to some obvious transformations)?   

Questions of this type have been studied previously; we list many results at the end of this section and a few especially relevant ones here.   The \emph{Minkowski
problem} asks whether a convex polygon exists that has the origin inside, has a given ordered set of edge 
lengths $\ell_0,\dots,\ell_{n-1}$,  and a given ordered set of \emph{(unit) edge normals} $\mathbf{n}_0,\dots,\mathbf{n}_{n-1}$, where $\mathbf{n}_i$ should be
the normal of the halfplane  supporting the edge of length $\ell_i$.
See Huang, Yang, and Zhang~\cite{huangMinkowskiProblemsGeometric2025}
for a history and a short proof that this 
is possible if and only if
$\sum_{i=0}^{n-1} \ell_i\mathbf{n}_i$ is the 0-vector. 

Our research was inspired by the open question asked by J. O'Rourke at CCCG'25~\cite{CCCG2025openproblems} (which in turn was inspired by%
~\cite{huangMinkowskiProblemsGeometric2025}): Can we reconstruct a polygon if we are given the edge normals 
and the areas of the \emph{triangles spanned by the edges},
i.e., the triangles formed by two consecutive corners and the origin?    
This problem was introduced by Stancu, who called it the 
\emph{$L_0$-Minkowski problem}~\cite{Stancu2002}.
Stancu showed that there always is a solution as long as no two edge normals
are opposite to each other.   
Stancu's proof is not constructive and implies no algorithm for computing the desired polygon.

We were not able to find a practical algorithm for the $L_0$-Minkowski problem, but 
inspired by these questions, we decided to explore the problem of reconstructing
a convex polygon from features more broadly.
Specifically, we looked
at six features that naturally arise in any convex polygon with the origin inside, and more generally in any polygon $P$ that is \emph{star-shaped from the origin},  i.e., all segments 
from the origin to a corner lie strictly inside $P$.   
(%
We simply call this ``star-shaped'' here.)
Say polygon $P$ has corners $c_0,\dots,c_{n-1}$ in counter-clockwise order. %
For $i\in \left\{0,\dots,n-1\right\}$, we define the following, treating all addition among indices modulo $n$ (see also Figure~\ref{fig:definitions}):
\pagebreak

\begin{figure}[ht]
  \centering
  \includegraphics[scale=1,page=2]{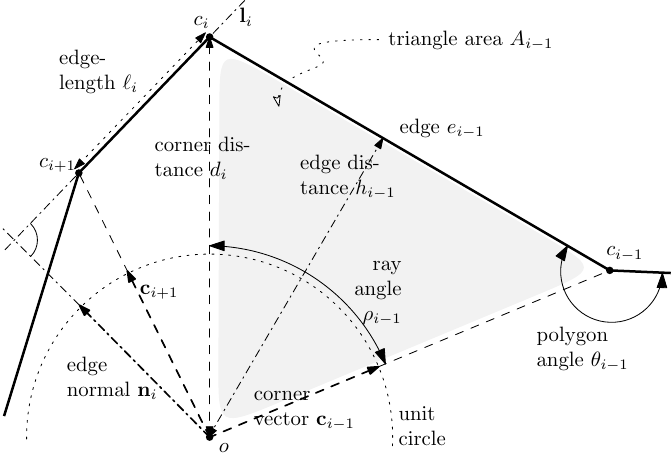}
  \caption{Features of a star-shaped polygon (bold).   
  }%
  \label{fig:definitions}%
  \label{fig:nomenclature}
\end{figure}

\begin{itemize}
\itemsep -2pt
\item The \emph{edge length} $\ell_i>0$ is the length of the \emph{edge} $e_i=\overline{c_{i}c_{i+1}}$.
\item The \emph{(triangle) area} $A_i$ is the area of the triangle $T_i$ spanned by $e_i$ and the origin.
\item The \emph{(unit) edge normal} $\mathbf{n}_i$ is the unit-length vector that is perpendicular to the line $\mathbf{l}_i$ supporting $e_i$, and that has a positive inner product with all points on~$\mathbf{l}_i$.%
\footnote{This is of course equivalent to any normal vector if we are in a computer model that can compute $\sqrt{\cdot}$, but requiring it to be a unit vector will allow us to use a simpler computer model sometimes.
We will not study any other normals and hence omit ``unit'' from now on.}   Edge normals are closely related to the \emph{polygon angle} $\theta_i$ (the inner angle at corner $c_i$):
$\theta_i$ equals $180^\circ$ %
minus the angle between $\mathbf{n}_{i-1}$ and  $\mathbf{n}_i$. 
The edge normal also corresponds to the \emph{edge ray}, which is the ray from origin $o$ along normal $\mathbf{n}_i$.

\item The \emph{edge distance} $\edgedistance_i>0$ is 
the distance of the line $\mathbf{l}_i$ supporting $e_i$ from the origin.   (The name ``$h$'' was chosen since this is the height of $T_i$ if we view $e_i$ as the base of this triangle.)
\item The \emph{(unit) corner vector} $\mathbf{c}_i$ is the unit-length vector along the \emph{corner ray}, i.e., from the origin along the ray through corner $c_i$.\footnote{As with edge normals, requiring unit length is helpful sometimes, but not required if we can compute $\sqrt{\cdot}$.}
Corner vectors are closely related to the  \emph{ray angle} $\rho_i$ (the angle between the corner rays): $\rho_i$ is the angle between $\mathbf{c}_i$ and $\mathbf{c}_{i+1}$.
\item The \emph{corner distance} $\cornerdistance_i$ is $||c_i||$, i.e., the distance of
	corner $c_i$ from the origin.
\end{itemize}

Each of these features consists of a sequence of $n$ numbers.
A polygon $P$ is defined by $2n$ coordinates, and so one would
expect that it takes two of these features 
to specify the polygon.    This is the main reason for us to consider the 
reconstruction question for all possible pairs of features.    Surprisingly
enough, the back-of-the-envelope computation of the degrees of freedom can be off: 
sometimes even specifying just one feature can make it impossible to have a convex polygon
(see, for example, Section~\ref{sec:edge_distance-edge_distance}).   For this reason,
we also study all the situations where only one feature has been specified.
Table~\ref{tab:settings} gives an overview of all the \emph{scenarios} (combinations of features) that we study, 
with links to the corresponding sections.

\begin{center}
\begin{table}[h!]
\centering
  \newcolumntype{V}{>{\centering\arraybackslash} m{1.235cm} }
  \newcommand{\none}{$\cdot$}
  \newcommand{\diag}{\backslash}
  \newcommand{\Block}{}
  \newcommand{\prob}[3]{\mbox{\rlap{\ref{sec:#2-#3}}\hyperref[sec:#2-#3]{\includegraphics[page=#1,scale=0.6]{figures/polygon_data_new}}}}
  \newcommand{\para}[2]{\Block{}{\includegraphics[page=#1,scale=0.5]{figures/polygon_data_new}}}
  \begin{tabular}{@{}c|VVVVVV}%
& triangle & edge & edge & corner & corner & edge \\
& area    & length & distance & vector & distance & normal \\
	& $A_i$                         
	& $\ell_i$                             
	& $\edgedistance_i$ 
	& $\mathbf{c}_i$                        
	& $\cornerdistance_i$                                
	& $\mathbf{n}_i$                       
	\\ 
\hline
$A_i$      
	& \prob{18}{area}{area}
	& \prob{1}{area}{edge_length}   
	& \prob{5}{area}{edge_distance} 
	& \prob{3}{area}{corner_vector}     
	& \prob{4}{area}{corner_distance}    
	& \prob{2}{area}{edge_normal}    
	\\
$\ell_i$   
	& \none %
	& \prob{19}{edge_length}{edge_length}
	& \prob{9}{edge_length}{edge_distance} 
	& \prob{7}{edge_length}{corner_vector}     
	& \prob{8}{edge_length}{corner_distance}    
	& \prob{6}{edge_length}{edge_normal}    
	\\
$\edgedistance_i$      
	& \none %
    & \none
	&  \prob{23}{edge_distance}{edge_distance} 
	& \prob{14}{corner_vector}{edge_distance} 
	& \prob{15}{corner_distance}{edge_distance} 
	& \prob{12}{edge_normal}{edge_distance} 
	\\
$\mathbf{c}_i$   
	& \none %
	& \none %
	& \none 
	& \prob{21}{corner_vector}{corner_vector}
	& \prob{13}{corner_vector}{corner_distance}    
	& \prob{10}{edge_normal}{corner_vector}     
	\\
$\cornerdistance_i$      
	& \none %
	& \none %
	& \none 
	& \none %
	& \prob{22}{corner_distance}{corner_distance}
	& \prob{11}{edge_normal}{corner_distance}    
	\\
$\mathbf{n}_i$
	& \none %
	& \none %
	& \none                                
	& \none                               
	& \none 
	& \prob{20}{edge_normal}{edge_normal}
	\\
  \end{tabular}
  \caption{Studied combinations of features.}%
  \label{tab:settings}
\end{table}
\end{center}

\paragraph{Our contribution.}
The main contribution of this work is the list of
reconstruction problems, together with a compilation about what is already known.   
Additionally, we investigate some of the less-studied scenarios.
Our results here include the following:
\begin{itemize}
\item In the scenario ``only triangle areas'', there always is a convex realization,
	and it can be found in linear time 
	(Section~\ref{sec:area-area}).
\item In the scenario ``only edge distances'', there is not always a convex realization, but there always is a star-shaped one, and
	we can test whether there exists a convex one in linear time if edge distances are distinct   
	(Section~\ref{sec:edge_distance-edge_distance}).
	Mostly due to polarity, the same holds for the scenario ``only corner distances''
	(Section~\ref{sec:corner_distance-corner_distance}).
\item In the scenario ``edge normals and corner vectors'', there is not always a convex realization or even a star-shaped realization.
The existence of such realizations can be tested in polynomial time (Section~\ref{sec:edge_normal-corner_vector}).
\item In the scenario ``areas and corner distances'', there is not always a convex or a star-shaped realization.
We can test in polynomial time whether there exists one, both for convex and for star-shaped 
(Section~\ref{sec:area-corner_distance}).
\item In the scenario ``areas and corner vectors'', there is not always a convex or a star-shaped realization. In both cases, we can test in polynomial time whether there exists one.
We also consider the variation where the features are given as an \emph{unordered} set, 
	and can argue that testing the existence of a polygon is then NP-hard even for a 
	star-shaped polygon
	(Section~\ref{sec:area-corner_vector}).
\end{itemize}

Some of these results are obtained by fixing one additional length and ``propagating'', i.e., arguing that with this all other features are fixed.   In some other scenarios, we argue that such propagation is difficult because there can be locally a choice of how to choose the next corner.   In fact, this can be exploited sometimes for hardness results, at least 
for some variants of the problem.   See Sections~\ref{sec:corner_distance-edge_distance} and~\ref{sec:edge_length-corner_vector} for details. 
The remaining scenarios, discussed in Sections~\ref{sec:area-edge_normal}--\ref{sec:area-edge_length}, are open.

\paragraph{Further Related Results.}%

The literature on reconstruction problems in the plane is quite rich, and there are naturally several results where the input data do not correspond to any of the features (or their combinations) considered in this paper.

Snoeyink~\cite{snoeyink_cross-ratios_1999} studied a reconstruction problem where we are given a polygon $P$ (with a triangulation), the polygon angles, and the so-called \emph{cross-ratios} (a ratio between edge-lengths at the quadrangles defined by the internal edges). He showed that one can uniquely reconstruct the polygon $P$, up to scaling and rotation, given this information. Biedl, Durocher and Snoeyink~\cite{biedl_reconstructing_2011} focused on a setup where the goal is to reconstruct a simple polygon from different types of data that can be obtained from range scanners, e.g., a set of points on the boundary of the polygon  such that each point records the line supporting the edge on which the point lies, or such a set that additionally records a vector perpendicular to the line that
points towards the polygon’s interior. Barequet, Gotsman and Sidlesky~\cite{DBLP:conf/cccg/BarequetGS06} considered the problem of reconstructing a planar polygon from its intersections with a collection of arbitrarily-oriented “cutting” lines.

Chen and Wang~\cite{Chen2012254} considered a polygon reconstruction problem where visibility information is relevant: In this case, the input consists of the sorted vertex sequence along the boundary and, for each vertex, the sequence of angles defined by all other visible vertices. Their reconstruction algorithm is an improvement with respect to the one by Disser, Mihalák and Widmayer~\cite{disser_polygon_2011}. A famous reconstruction problem involving visibility is that of reconstructing a polygon from its vertex visibility graph. The version closest to our paper is probably the one studied by Coullard and Lubiw~\cite{DBLP:journals/ijcga/CoullardL92}:  Given an edge-weighted graph $G$, is it the visibility graph of a simple polygon with the given weights as Euclidean distances? See also the Computational Geometry Handbook~\cite[Chapter 34]{goodman_handbook_2018} for other polygon reconstruction problems.

Finally we mention that there are many related reconstruction problems in 3D; we refer to a recent paper by Cho and Kim~\cite{ChoKim25} that gives a good overview of the history of 3D problems.

\section{Preliminaries}

We assume familiarity with basic notation around \emph{polygon}, \emph{convex}, and \emph{star-shaped}.
When reconstructing a polygon, we always demand that the origin is strictly inside, and
disallow \emph{flat} ($180^\circ$) polygon angles. 
Sometimes we consider a convex polygon $P$ that additionally has 
\emph{edge rays hitting their edges}, 
i.e., for all $i$ the ray along edge normal $\mathbf{n}_i$ actually intersects $e_i$.
(This is violated in Figure~\ref{fig:nomenclature} at $\mathbf{n}_i$.) 

We often use the well known \emph{side-area-side-formula} (or \emph{SAS-formula}), which in our notation means that $A_i=\tfrac{1}{2}d_i d_{i+1}\sin 
\rho_{i}$ 
for all indices $i$. 
Also, $A_i=\ell_i h_i/2$.

The following will be a helpful tool to reduce the number of scenarios.
Any convex polygon $P$ with the origin strictly inside defines another convex polygon $P^\polar$ (the \emph{polar polygon with respect to the unit circle}) as follows.    Map every corner $(c_x,c_y)$ to the half-plane $\{(x,y): c_x x+c_y y\leq 1\}$, and let $P^\polar$ be the intersection of these half-planes.   From the definition it follows, that the corner rays of $P$ equal the edge rays of $P^\polar$, and furthermore, the corner distances of $P$ and the edge distances of    $P^\polar$ are inverse to each other. %
The polar of the polar polygon is the original~\cite{grunbaum2003convex}, 
and as such, the edge distances of $P$ also equal the inverse corner distances of $P^\polar$. 
The polar polygon is not well-defined if $P$ is not convex, see also Inchbald~\cite{Inchbald24} for related discussions.

\paragraph{Further Features.}%
\label{app:tangentcircle}%
\label{app:partialrayangle}%
\label{app:partialtriangle}
Among the features not listed in the introduction are a few more that are derived (directly or via combinations) from the earlier features.  As these will come in handy later, we define them here.   For $i\in \{0,\dots,n{-}1\}$, we define the following (see also Figure~\ref{fig:definitions_more}):
\begin{itemize}
\item  The \emph{tangent circle} $\tangentcircle_i$ is the circle centered at the origin and with radius $h_i$.   Thus, the line $\mathbf{l}_i$ that supports edge $e_i$ must be tangent to $\tangentcircle_i$.
\item The \emph{tangent point} $x_i$ is the point where $\mathbf{l}_i$ is tangent to $\tangentcircle_i$.   Equivalently, it is the point where the edge ray of $e_i$ hits the line supporting $e_i$; this point may or may not lie on edge $e_i$.
\item The \emph{(backward) partial triangle} $T_i^-$ is the triangle spanned by the origin $o$, the corner $c_i$, and the tangent point $x_i$.   Put differently, it is defined by the edge ray, the corner ray, and the line through $e_i$.    Its angle at $o$ is called the
\emph{(backward) partial ray angle} and denoted $\rho_i^-$.   Equivalently, $\rho_i^-$ is  the angle between $\mathbf{n}_i$ and $\mathbf{c}_{i}$ that lies in the range $[0,90^\circ)$.   

\item The \emph{(forward) partial triangle} $T_i^+$ and \emph{(forward) partial ray angle} $\rho_i^+$ are defined symmetrically, except that they use corner $c_{i+1}$.   
\end{itemize}
Note that if the edge ray does not hit the edge, then one of the two partial triangles is a subset of the other, so the triangle $T_i$ may be their union or their symmetric difference (and similarly the ray angle $\rho_i$ is not necessarily $\rho_i^-+\rho_i^+$).   

\begin{figure}[t]
  \centering
  \includegraphics[scale=1,page=3]{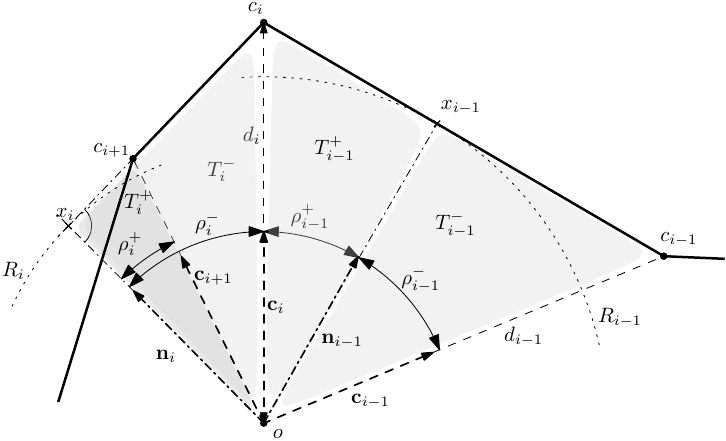}
  \caption{More features of a star-shaped polygon. }%
  \label{fig:definitions_more}
\end{figure}

\paragraph{Input and computer model.}
Whenever we speak of ``given input'', we assume that it satisfies some obviously necessary conditions
for there to be a star-shaped polygon.
In particular, all areas and lengths must be strictly positive, corner vectors must not all lie in one half-space, and likewise for the edge normals. Also,  corner vectors (and for convex
polygons also the edge normals) must appear in the circular ccw order around the origin.
Unless explicitly stated otherwise, we assume that the input is labelled, i.e., we know exactly which specified
length/distance/ray/triangle-area corresponds to which edge/corner.

As usual, we assume that the input has been specified with numbers that are small enough to fit
into a computer word, and we can do basic arithmetic operations on computer words in constant time
(the Word\-RAM model).
This does not permit square roots or
trigonometric functions to be computed in constant time,
so occasionally we need 
the more powerful RealRAM model (which permits these~\cite[Section 2.3]{ShamosThesis}).

\section{Settings Defined by One Type of Information}%
\label{sec:one_feature}

We begin 
with the scenarios where only one feature is specified. Many results are known here already, but to our surprise some appear to be unstudied.

\subsection{Known results}

\begin{enumerate}
    \item%
\makeatletter\edef\@currentlabel{\arabic{section}.\arabic{subsection}.\arabic{enumi}}\makeatother%
\label{sec:edge_normal-edge_normal}
Consider the scenario ``only edge normals'', which is equivalent to being given polygon angles (up to rotation).
Reconstructing a polygon from its polygon angles 
appears to first have been studied by
Vijayan~\cite{Vijayan1986}, who proved that any sequence of polygon angles can be realized as long as 
they sum to $(n-2)180^\circ$. 
See also~\cite{disser_polygon_2011} and~\cite{efrat_polygons_2022} for further results.

\item%
\makeatletter\edef\@currentlabel{\arabic{section}.\arabic{subsection}.\arabic{enumi}}\makeatother%
\label{sec:corner_vector-corner_vector}
The scenario ``only corner vectors'' is the same as the previous for a convex polygon by polarity.   Any set of corner vectors can be realized by a star-shaped polygon.
\item%
\makeatletter\edef\@currentlabel{\arabic{section}.\arabic{subsection}.\arabic{enumi}}\makeatother%
\label{sec:edge_length-edge_length}

The scenario ``only edge lengths'' is closely related to
\emph{linkage reconfiguration} problems. Lenhart and Whitesides~\cite{LenhartW95} showed that edge lengths can be realized if and only if the longest edge-length is less than the sum of all others. If so, then the edge lengths can even be realized by a convex polygon  with all vertices on a common circle~\cite{cyclic_polygons}.

\end{enumerate}

\subsection{Triangle Area}%
\label{sec:area-area}

The triangle area is less obvious than other features, and so it is perhaps not surprising that this scenario appears to not have been studied before.
We show that there always exists a realization.

\begin{lemma}
Let $A_0,\dots,A_{n-1}$ be an arbitrary sequence of positive numbers.   Then there exists an $n$-vertex convex polygon $P$ with the origin inside such that the triangles defined by consecutive corners and the origin have areas $A_0,\dots,A_{n-1}$, in this order.
\end{lemma}
\begin{proof}
    Let $A_{\max}$ be the maximum among the given areas.   Choose $r$ such that 
    $ r^2 > 2A_{\max}/{\sin{\tfrac{90^\circ}{n-2}}}$.
Then, fix a circle $C$ of radius $r$ centered at the origin, place ${c}_1$ at $(r,0)$, and for $i=2,\dots,n-1$, place ${c}_i$ on $C$ such that 
triangle $T_i$ has area $A_i$; see also Figure~\ref{fig:only_area}.  By the SAS-formula 
$\sin \rho_i  =2A_i/r^2 \leq 2A_{\max}/r^2 < \sin(\frac{90^\circ}{n-2})$, and so ray angle $\rho_i$ is the solution to $\arcsin(2A_i/r^2)$ with  $\rho_i<\frac{90^\circ}{n-2}$.

\begin{figure}[ht]
\centering
    \includegraphics[page=3,scale=1,trim = 50 50 0 0,clip]{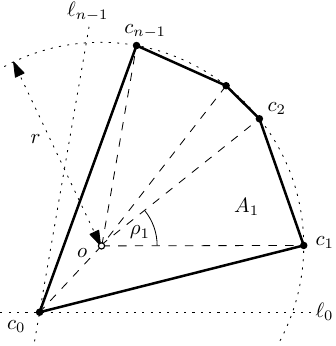}
    \caption{A convex polygon that realizes given areas. }%
\label{fig:only_area}
\end{figure}

The convex hull of corners ${c}_1,\dots,{c}_{n-1}$ and the origin $o$ forms a fan $F$, and its angle at $o$ is
$\sum_{i=1}^{n-2} \rho_i < 90^\circ$ since $\rho_i<\frac{90^\circ}{n-2}$.
It remains to choose a point $c_0$ to add two triangles to $F$ that realize areas $A_0$ and $A_{n-1}$.   
    Define line ${\ell}_0$ to be parallel to $\overline{oc_1}$ and of distance $2A_0/r$ to it on the side opposite to $c_2$;    placing $c_0$ anywhere on $\ell_0$ would give a triangle $oc_0c_1$ of area $A_0$.
    Similarly define a line $\ell_{n-1}$ parallel to $\overline{oc_{n-1}}$ of distance $2A_{n-1}/r$ on the side opposite to $c_{n-2}$.   
    Using the intersection point of these lines as $c_0$ gives convex angles at $c_0,c_1,c_{n-1}$,
which gives the desired convex polygon.
    \end{proof}

    Note that the proof is algorithmic, and one could easily construct the corresponding polygon in linear time in the RealRAM model.

\subsection{Edge Distance}%
\label{sec:edge_distance-edge_distance}

Consider the scenario where we are only given the edge distances
$\edgedistance_0,\dots,\edgedistance_{n-1}$.

We begin by noting a few easy observations about some of the features in a star-shaped polygon.   Fix one $i\in \{0,\dots,n{-}1\}$.

\begin{observation}%
\label{obs:backward_ray_angle}
Backward ray angle
$\rho_i^-$ is the solution of $\arccos(h_i/d_{i})$ that lies in the range $[0,90^\circ)$.
Forward ray angle
$\rho_i^+$ is the solution of $\arccos(h_i/d_{i+1})$ that lies in the range $[0,90^\circ)$.
\end{observation}

\begin{claim}
We have $d_{i+1}\geq \max\{h_{i+1},h_i\}$.
\end{claim}
\begin{proof}
The corner $c_{i+1}$ has distance $d_{i+1}$ from the origin, and lies on two lines tangent to the circles $\tangentcircle_{i+1}$ and $\tangentcircle_i$ of radius $h_{i+1}$ and $h_i$.
This is possible only if it is not inside these circles.   
\end{proof}

\begin{corollary}%
\label{claim:lower_bound_corner_vector}%
\label{claim:big_ray_angle}
The partial ray angle $\rho_i^{-}$ and $\rho_i^{+}$ lie in $[0,90^\circ)$ and respectively satisfy 
$\cos \rho_i^{-}\leq h_i/h_{i-1}$ and
$\cos \rho_i^{+}\leq h_i/h_{i+1}$.
\end{corollary}
\begin{proof}   
This is implied by the two previous claims since $\cos$ is monotonically decreasing in the range $[0,90^\circ)$.
\end{proof}

The following theorem (which rephrases $h_0,\dots,h_{n-1}$ as radii of concentric circles) shows that
there always exists a star-shaped realization.

\begin{restatable}{thm}{OnlyEdgeDistancesStarshaped}%
\label{thm:normals_starshaped}%
\label{thm:only_edge_distances_starshaped}
For any given sequence of $n$ concentric circles $\tangentcircle_0,\dots,\tangentcircle_{n-1}$ there exists a star-shaped polygon $P$ with edges $e_0,\dots,e_{n-1}$ such that  for all $i$ the supporting line 
of $e_i$ is tangent to $\tangentcircle_i$.
\end{restatable}
\begin{proof} The proof is by induction on $n$.
We prove the stronger claim that 
all vertices of $P$ lie strictly outside of all circles. For the base case $n=3$, assume that $R_0$ has maximum radius.
A vertical line tangent to  $\tangentcircle_0$ on the right, 
a horizontal line tangent to $\tangentcircle_2$ on the bottom and a slanted
line tangent to $\tangentcircle_1$ clearly suffice to define $P$ as desired 
if we choose the slope of the slanted line very small or large as needed to get all vertices outside of $R_0$.
See \cref{fig:tangent}.

Assume now that the theorem holds for $n=k-1$. 
If all radii are the same, then we draw the polygon as a regular polygon
and scale it such that all edges are tangent to the common circle.
Otherwise, we may assume
that $\tangentcircle_1$ is one of the circles with minimal 
radius and that the radii of  $\tangentcircle_0$, $\tangentcircle_1$ and $\tangentcircle_2$ are not all the same. We also may assume that the radius of $\tangentcircle_0$ is not smaller than
 the radius of $\tangentcircle_2$,
otherwise we reverse the order of circles and mirror the polygon in the end. 
Consequently, the radius of $\tangentcircle_0$ is strictly larger 
than the radius of $\tangentcircle_1$.
By the induction hypothesis, we can realize a star-shaped polygon $P'$,
whose edges are tangent to $\tangentcircle_0,\tangentcircle_2,\ldots, \tangentcircle_{k-1}$.
We show how to modify $P'$ to get to $P$. 
To add an edge $e_1$ for the omitted circle $\tangentcircle_1$, we look at the adjacent edges $e_{0}$ and $e_{2}$ in $P'$ and their common endpoint $p$. 
There are two cases based on whether the corner at $p$ is convex or reflex; see Figure~\ref{fig:tangent}. 
In both cases, we shorten $e_0$ (to a new corner $r$), lengthen $e_2$ (to a new corner $q$), and ensure
that $\overline{rq}$ lies on a line $\mathbf{l}_1$ tangent to $\tangentcircle_1$ with the origin on the correct side
for star-shapedness.    More specifically, pick a point $r$ in the interior of $e_0$, but close enough to $p$
so that it (like $p$) is outside all circles.    Let $\mathbf{l}_1$ be the line tangent to $\tangentcircle_1$
and through $r$; in the convex case we take the one that has the origin and $p$ on the same side, while in
the reflex case we take the one that separates the origin from $p$.   
Note that this tangent is
different from the supporting line of $e_0$, since $\tangentcircle_1$ has strictly smaller radius than $\tangentcircle_0$.
Also, the intersection point $q$ between $\mathbf{l}_1$ and the supporting line of $e_2$ is arbitrarily close to $p$ if we pick $r$ sufficiently close to $p$; 
in particular we can pick 
$r$ such that $q$ is outside all circles as well.    One verifies that in both cases
the polygon obtained by re-routing via $r$ and $q$ is star-shaped, see also Figure~\ref{fig:tangent}.
With this the induction step is proved.
\end{proof}

\begin{figure}
    \centering
    \includegraphics[scale=1, page=1]{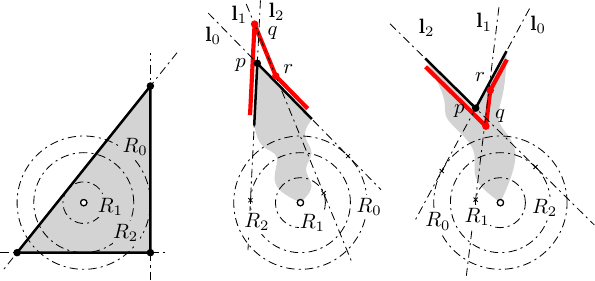}
    \caption{Inductive construction of Theorem~\ref{thm:normals_starshaped}.  Left: Base case. Middle/Right: Schematic inductive step at a convex/reflex corner of $P$. }%
    \label{fig:tangent}
\end{figure}

Convex realizations do not always exist, but we can 
characterize when they exist, presuming edge distances are distinct.  
To state the characterization, we first need some definitions.
Fix an instance $\calI$, i.e., a circular list $h_0,\dots,h_{n-1}$ of distinct edge distances. 
A \emph{local minimum [maximum]} is an edge distance that is smaller [larger] than both its predecessor and its successor in this circular list.   
If there are at least two local minima, then 
the \emph{reduced instance} $\calI'$
consists of all those edge distances (in the induced order) that are local minima or local maxima of $\calI$.   If there is only one local minimum, then $\calI'$ consists of it, the (unique) local maximum, and one edge distance between the two.   
Enumerate $\calI'$ as $h'_0,\dots,h'_{k-1}$, where $h'_0$ is the global minimum of the edge distances.
The case $|\calI'|=3$ is a special case handled separately.   For $|\calI'|\geq 4$, we have $k$ even and the items in $\calI'$ alternate between local minimum and local maximum.
We now have the following equivalences.

\begin{restatable}{thm}{OnlyEdgeDistancesConvex}%
\label{thm:only_edge_distances_convex}
The following are equivalent for an instance $\calI=(h_0,\dots,h_{n-1})$ 
of the problem of realizing distinct edge distances:
\begin{enumerate}
\itemsep -2pt
\item $\calI$ has a convex realization.%
\label{it:I_convex}
\item The reduced instance $\calI'$ has a convex realization.%
\label{it:Iprime_convex}
\item $\calI'=(h'_0,\dots,h'_{k-1})$ either has size 3, or satisfies
	$\sum_{j=0}^{k/2-1} (\psi_{2j}^{+1}+\psi_{2j}^{-1}) < 360^\circ$
	where $\psi_{2j}^{\ell}$  (for $\ell=\pm 1$)
	is the solution to $\arccos(h'_{2j}/h'_{2j+\ell})$ in $(0,90^\circ)$.%
\label{it:angle_sum}
\item $\calI'$ has a convex realization where edge rays hit the interior of their edges.%
\label{it:Iprime_RayHitEdge}
\end{enumerate}
\end{restatable}

Before proving the theorem, we note that Condition~(\ref{it:angle_sum}) is easily tested in linear time, as long as the computer model supports the computation of $\arccos$.
Also, it is easy to construct instances where this condition is violated:
Let $\calI$ be an instance with $n\geq 6$ edge distances alternating between more than 1 and less than $\tfrac{1}{2}$.   Every edge distance is a local minimum or maximum, and $\psi_{2j}^\ell>60^\circ$ for all indices $2j$ and $\ell$, so the angle-sum in (\ref{it:angle_sum}) exceeds $360^\circ$ for $n\geq 6$.  

Let us prove a couple of observations needed for the proof of Theorem~\ref{thm:only_edge_distances_convex}.

\begin{claim}%
\label[claim]{claim:RayHitEdge}
If a sequence of edge distances is realizable by a convex polygon, then for any index $i$ with 
$h_i\leq \min\{h_{i-1},h_{i+1}\}$,     the edge ray through $\mathbf{n}_i$ hits $e_i$.
\end{claim}
\begin{proof}
Up to rotation we may assume that $\mathbf{n}_i$ goes vertically upward from $o$, 
so the tangent point $x_i$ is $(0,h_i)$.
See also Figure~\ref{fig:normal-in-edge-two-invalid-choices}.
Assume for contradiction that 
$e_i$ lies strictly left of $(0,h_i)$ (the case of `strictly right' is symmetric).
Then $c_{i}$ lies to the left of $(0,h_i)$.   It also lies on $\mathbf{l}_{i-1}$, which is tangent to circle $\tangentcircle_{i-1}$, which encloses $\tangentcircle_i$ (or is identical to it) by $h_i\leq h_{i-1}$.    There are two choices for such a tangent-line, defining two possible ways in which $e_{i-1}$ arrives at $c_i$ since the origin must be to the left of $e_{i-1}$ when walking from $c_{i-1}$ to $c_{i}$.  %
One verifies that both of those lead to an angle that is at least $180^\circ$ at $c_{i}$.   
\end{proof}

\begin{figure}[ht]
    \centering{}
    \includegraphics[page=2,trim=0 110 0 0,clip]{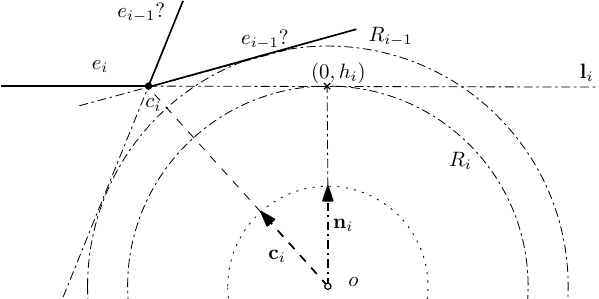}
    \caption{At an edge with locally minimal edge distance, the edge ray must hit the edge.}%
    \label{fig:normal-in-edge-two-invalid-choices}
\end{figure}

\begin{lemma}%
\label[lemma]{lemma:small_n}
For $n=3,4$, any set of edge distances can be realized with a convex polygon where edge rays intersect the interior of their corresponding edges.
\end{lemma}
\begin{proof}
    For $n=3$, the base case of Theorem~\ref{thm:normals_starshaped} creates a suitable convex polygon.   For $n=4$, consider the four tangent-circles, and add tangent-lines that are alternatingly horizontal and vertical, and are below/right/above/left of the origin.
The intersection of the corresponding halfplanes defines a suitable realization with a rectangle, see Figure~\ref{fig:edge_distance_n=4}.
\end{proof}

\begin{figure}[ht]
\hspace*{\fill}
\includegraphics[scale=1,page=4,trim=0 0 0 0 ,clip]{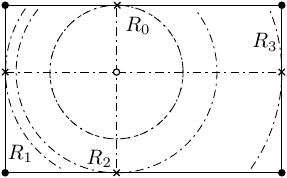}
\hspace*{\fill}
\caption{Constructing a rectangle for any four given edge distances.}%
\label{fig:edge_distance_n=4}
\end{figure}

\begin{proof}[Proof of Theorem~\ref{thm:only_edge_distances_convex}]
\paragraph{(\ref{it:I_convex})$\Rightarrow$(\ref{it:Iprime_convex}).}   This implication at first glance seems trivial, since
the edge distances for $\calI'$ are a subsequence of the ones for $\calI$.   Let $H_0,\dots,H_{n-1}$ be the
half-planes that define a realization $P$ of $\calI$, i.e., $H_i=\{p: \langle \mathbf{n}_i, p \rangle\leq h_i\}$.
Let $P'=\bigcap_{i\in I} H_i$, where $I$ is the set of indices in $\calI'$.    Every edge $e_i$ with $i\in I$ is
part of an edge in $P'$, and has the same edge distance, so $P'$ looks like a realization of $\calI'$.   The only problem 
is that $P'$ may be unbounded, or equivalently, the edge normals of $P'$ (enumerate them as $\mathbf{n}'_0,\dots,\mathbf{n}'_{k-1}$) all lie in one half-space.
See also Figure~\ref{fig:onlyEdgeDistancesConvexCloseUp}.

\begin{figure}[ht] \centering
\includegraphics[page=7,trim=0 0 0 0,clip]{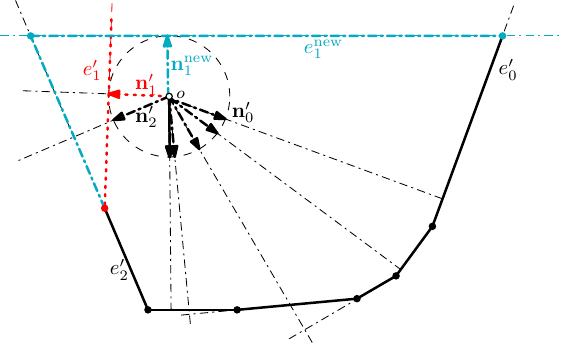}
\caption{Moving the normal for one half-plane to make their intersection bounded.}%
\label{fig:onlyEdgeDistancesConvexCloseUp}
\end{figure}

So assume that (up to renaming) the ccw angle
between $\mathbf{n}'_0$ and $\mathbf{n}'_1$ is at least $180^\circ$.
We have $k\geq 3$, so $\mathbf{n}'_2$ is different from $\mathbf{n}'_0$.    
Define $\mathbf{n}^{\rm new}_1:=
-\tfrac{1}{2}(\mathbf{n}'_0+\mathbf{n}'_2)$ and note that this lies between 
$\mathbf{n}'_0$ and $\mathbf{n}'_2$ in ccw order and forms angles less than $180^\circ$ with both.  If we hence remove from $P'$ the half-plane defined by $\mathbf{n}'_1$ and add to it the half-plane defined by $\mathbf{n}^{\rm new}_1$ and the corresponding edge distance, then the intersection of these half-spaces is a realization of $\calI'$.

\medskip

\paragraph{(\ref{it:Iprime_convex})$\Rightarrow$(\ref{it:angle_sum}).}   
Assume that $k> 3$ (so $k$ is even since we have alternating local minima and maxima), and  fix a convex realization of $\calI'$.
At every local minimum $h'_{{2j}}$ the edge ray hits edge $e_{{2j}}$ by \cref{claim:RayHitEdge}, and hence the ray angle is
$\rho_{{2j}}=\rho_{{2j}}^{+}+\rho_{{2j}}^{-}$.
We know that
$\rho_{{2j}}^{+}\geq \psi_{{2j}}^{+1}$  and
$\rho_{{2j}}^{-}\geq \psi_{{2j}}^{-1}$   
by Corollary~\ref{claim:lower_bound_corner_vector}.
Hence, $\rho_{{2j}}\geq \psi_{{2j}}^{+} +\psi_{{2j}}^{-}$   
and $\sum_{j=0}^{k/2-1} \rho_{{2j}}\geq \sum_{j=0}^{k/2-1} (\psi_{{2j}}^{+1} +\psi_{{2j}}^{-1})$.   
On the other hand $\sum_{j=0}^{k/2-1} \rho_{{2j}}<360^\circ$  since the sum does not count the ray angles corresponding to local maxima, which are non-zero.

\medskip

\paragraph{(\ref{it:angle_sum})$\Rightarrow$(\ref{it:Iprime_RayHitEdge}).}     The claim is trivial for $k=3,4$ following \cref{lemma:small_n}, so we may assume that $k>4$.
Since 
$k$ is even, we actually have $k\geq 6$.  
For $\varepsilon>0$, $\ell\in \{+1,-1\}$ and $j\in \{0,\dots,k/2 - 1\}$, define 
$$\rho_{2j}^\ell(\varepsilon):=\arccos \left(\frac{\edgedistance'_{2j}}{\edgedistance'_{2j+\ell}+\varepsilon}\right),$$ 
where we mean the solution in $(0,90^\circ)$.
In other words, $\rho_{2j}^\ell(\varepsilon)$ is the angle of a right triangle $T_{2j}^\ell$ where the adjacent side has length $\edgedistance'_{2j}$ and the hypotenuse has length $\edgedistance'_{2j+\ell}+\varepsilon$, see also 
Figure~\ref{fig:onlyEdgeDistancesConvexConstruct}.
Angle $\rho_{2j}^\ell(\varepsilon)$ increases monotonically with $\varepsilon$, and $\rho_{2j}^\ell(0)=\psi_{2j}^\ell$ while $\rho_{2j}^\ell(\varepsilon)\rightarrow 90^\circ$ as $\varepsilon\rightarrow \infty$.   
Since  $\sum_{j=0}^{k/2-1} (\psi_{2j}^{+1}+\psi_{2j}^{-1})$ is strictly less than $360^\circ$, 
while we have $k\geq 6$ such angles in total, there exists an $\varepsilon>0$
such that 
$\sum_{j=0}^{k/2-1} (\rho_{2j}^{+1}(\varepsilon)+\rho_{2j}^{-1}(\varepsilon))=360^\circ$, 
and we use this $\varepsilon$ for the construction (and drop it from now on).

\begin{figure}[ht] \centering
\includegraphics[page=2]{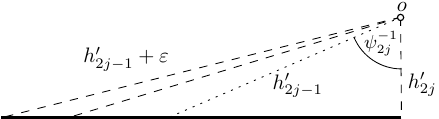}
\caption{We can choose $\varepsilon$ such that the sum of partial ray angles is exactly $360^\circ$.}%
\label{fig:onlyEdgeDistancesConvexConstruct}
\end{figure}

Place triangles $T_{0}^{-},T_{0}^{+},T_{2}^{-},\dots,T_{k-2}^{+}$ so that they share the corner $o$, and such that $T_{2j}^{-}$ and $T_{2j}^{+}$ share the side of length $h'_{2j}$ while $T_{2j}^{+}$ and $T_{2j+2}^{-}$ share the side of length $h'_{2j+1}+\varepsilon$.   See also Figure~\ref{fig:onlyEdgeDistancesConvexMerge}.
By our choice of $\varepsilon$ this fills the sector around $o$ exactly, and the opposite sides of $T_{2j}^{-}$ and $T_{2j}^{+}$ become the
edge $e_{2j}$ of a convex polygon $P''$ that 
realizes $h'_0,h'_2,\dots,h'_{k-2}$.
The edge ray of $e_{2j}$ goes along the side shared by $T_{2j}^{-}$ and $T_{2j}^{+}$, and in particular therefore hits $e_{2j}$.

\begin{figure}[ht] \centering
\includegraphics[page=3,scale=0.7]{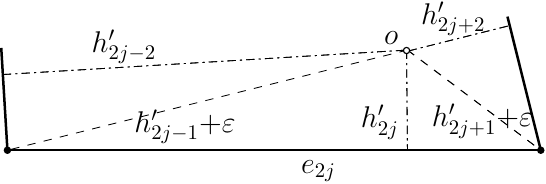}
\caption{Combining triangles to get a realization of the local minima.}%
\label{fig:onlyEdgeDistancesConvexMerge}
\end{figure}

It remains to create edges for the local maxima $\edgedistance'_1,\edgedistance'_3,\dots,\edgedistance'_{k-1}$.
See Figure~\ref{fig:onlyEdgeDistancesConvexOdd}.
For $i=0,\dots,k/2 - 1$, the common end $c$ of edges $e_{2j}$ and $e_{2j+2}$ has 
distance $\edgedistance'_{2j+1}+\varepsilon$ from the origin; in particular the corner ray of $c$ intersects first the two tangent-circles $\tangentcircle_{2j}$ and $\tangentcircle_{2j+2}$ (in some order), then it intersects $\tangentcircle_{2j+1}$ (since $h'_{2j+1}$ is a local maximum), and \emph{then} it intersects $c$.   
Let $x$ be the place where this ray intersects $\tangentcircle_{2j+1}$, and cut off $c$ with a line perpendicular to the ray and through $x$ (hence tangent to $\tangentcircle_{2j+1}$).   Since $x$ is outside $\tangentcircle_{2j},\tangentcircle_{2j+2}$, the tangent points of $e_{2j}$ and $e_{2j+2}$ remain part of the polygon.
Also, the edge rays of $e_{2j}$ and $e_{2j+2}$ separate the new edge from all other corners of $P''$, hence we can repeat the operation at all corners of $P''$ to get the desired convex realization of $\calI'$.

\begin{figure}[ht] \centering
\includegraphics[page=4,trim=250 45 0 115,clip]{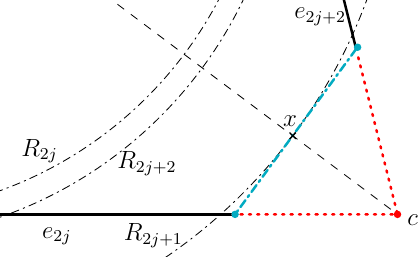}
\caption{Closeup: insert an edge for a local maximum.
}%
\label{fig:onlyEdgeDistancesConvexOdd}
\end{figure}

\medskip
\paragraph{(\ref{it:Iprime_RayHitEdge})$\Rightarrow$(\ref{it:I_convex}).}
To go from a realization $P'$ of $\calI'$ back to one of $\calI$, we have to re-insert edges whose distances were neither local minima nor local maxima.

Say we have a local minimum $h_0'$ and a local maximum $h_1'$ and in the original instance had between them $\ell$ edge distances
$h_{0}'<h_1<h_2<\dots<h_\ell<h_1'$.
(Recall that all distances are distinct, which is crucial in this construction.)
So we must add a new edge $e_i$ (for $i=1,\dots,\ell$) between the two edges $e_0'$ and $e_1'$ that realized edge distance $h_0'$ and $h_1'$.
Consider Figure~\ref{fig:onlyEdgeDistancesConvexNotMinMax}.
Let $c_1'$ be the corner between  $e'_{0}$ and $e'_1$, and
$x_{0}$ and $x_1$ be the tangent points of these edges.
Imagine sliding a segment $e=\overline{ab}$
where initially $a=x_{0}$ and $b=c'_1$, at the end $a=c'_1$ and $b=x_1$, and each point moves at constant speed.    

\begin{figure}[ht] \centering
\includegraphics[page=5,scale=0.8,trim=0 5 0 110,clip]{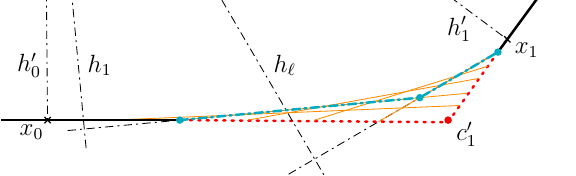}
\caption{Closeup: insert edges between a local minimum and a local maximum.}%
\label{fig:onlyEdgeDistancesConvexNotMinMax}
\end{figure}

When we begin the slide of $e$, 
the supporting line $\mathbf{l}$ of $e$ 
contains $e_0'$ and hence has distance $h'_0<h_1$ from the origin.
At the end of the slide the supporting line $\mathbf{l}$ has distance $h'_1>h_\ell$.
Since the movement is continuous,  there is a first time when $\mathbf{l}$ has distance $h_1$ from $o$.   
Cut off corner $c_1'$ with a halfplane through the current $\mathbf{l}$ to obtain an edge with distance $h_1$.   Repeat,
i.e., continue the slide of $e$ until $\mathbf{l}$ has distance $h_2$ from the origin, cut off part of the polygon with the current $\mathbf{l}$ to obtain a new edge of distance $h_2$, etc.
The last edge is added strictly before the slide ends by $h_\ell<h'_1$.   So in particular both $x_0$ and $x_1$ remain as points on $e'_0$ and $e'_1$.
Therefore, the edges that we added here do not interfere with any edges that we might add for some other minimum-maximum pair elsewhere.
Hence, 
repeating this approach at all other such pairs gives $(\ref{it:I_convex})$.

We note that this construction crucially 
needs distinct edge distances; it is not clear how one could insert two equidistant edges if flat angles are not allowed.
Also note that in this construction the new edges are not necessarily hit by edge-normals.   This in fact cannot be avoided:    There are instances that can be realized, but not when all edge rays must  hit their edges.  To see this, consider $\calI=(1,2,100,1.5,101)$. If we could realize this with edge rays hitting edges, then the total sum of partial ray angles has to equal $360^\circ$, since all of them contribute positively to the ray angles.
But due to \cref{claim:lower_bound_corner_vector} any realization would need four partial ray angles $\geq \arccos(2/100)>88^\circ$, and one partial ray angle $\geq \arccos(1/2)=60^\circ$.   So no realization with all edge rays hitting edges exists.   
On the other hand, there exists a convex realization for this (and actually for \emph{all} sets of $n=5$ edge distances), because the reduced instance has size at most 4 and Condition (\ref{it:angle_sum}) hence always holds.
\end{proof}

\subsection{Corner Distance}%
\label{sec:corner_distance-corner_distance}

For the scenario ``only corner distances'', 
we can test whether there exists a convex realization when all corner distances are distinct due to Theorem~\ref{thm:only_edge_distances_convex}  and using polarity.
There always exists a star-shaped realization, since we can simply place corners of the appropriate distance from the origin $o$, with ray angles $360^\circ/n$, and connect them in order.

\section{Settings Defined by Two Types of Information}%
\label{sec:two_features}

\subsection{Existing or Trivial Results}%
\label{sec:trivial2}

Quite a few of the scenarios where two types of features are given have either been studied before, or are trivial.   
\begin{enumerate}
\itemsep -2pt
    \item 
Reconstructing a polygon from edge lengths and edge normals
is the Minkowski Problem discussed in the introduction.%
\makeatletter\edef\@currentlabel{\arabic{section}.\arabic{subsection}.\arabic{enumi}}\makeatother%
\label{sec:edge_length-edge_normal}

\item%
\makeatletter\edef\@currentlabel{\arabic{section}.\arabic{subsection}.\arabic{enumi}}\makeatother%
\label{sec:corner_vector-corner_distance}
Reconstructing from corner vectors and corner distances is trivial, because the given information specifies the corners uniquely.

\item%
\makeatletter\edef\@currentlabel{\arabic{section}.\arabic{subsection}.\arabic{enumi}}\makeatother%
\label{sec:edge_normal-edge_distance}
Reconstructing from edge normals and edge distances is trivial, since these specify the lines supporting the edges, and each intersection of two consecutive of these lines defines the corner.

\item%
\makeatletter\edef\@currentlabel{\arabic{section}.\arabic{subsection}.\arabic{enumi}}\makeatother%
\label{sec:edge_length-corner_distance}

Reconstructing from edge lengths and corner distances is trivial, since all three side-lengths of triangle $T_i$ are given. So $T_i$ is unique for all $i$, and with this so is the realization up to rotation.
\end{enumerate}

\subsection{Edge Normal and Corner Vector}%
\label{sec:edge_normal-corner_vector}

Assume that we are given the sequences of edge normals $\mathbf{n}_0,\dots,\mathbf{n}_{n-1}$ and corner vectors $\mathbf{c}_0,\dots,\mathbf{c}_{n-1}$.
Recall that 
this determines the polygon angle $\theta_i$, which is $180^\circ$ minus the ccw angle
$\alpha_i$ between $\mathbf{n}_{i-1}$ and $\mathbf{n}_i$.

As such, the reconstruction problem in this scenario is very closely related to the so-called
Aleksandrov Problem, which asks to reconstruct a polygon $P$ from its corner vectors and polygon angles.
Necessary and sufficient conditions (depending on $\alpha_0,\dots,\alpha_{n-1}$ and $\mathbf{c}_0,\dots,\mathbf{c}_{n-1}$) are known for realizability,
see Huang~\cite{huangMinkowskiProblemsGeometric2025} and the references therein.
However, Alexandrov's problem is not exactly the same as ours, because the polygon angles determine the edge normals only up to rotation.
So while the conditions for Alexandrov's problem are certainly necessary, it is not clear whether they are also sufficient, i.e., we
can reconstruct a polygon $P$ that additionally has the corners exactly on the corner rays.   (We cannot simply rotate $P$ to achieve
this since this might not match the edge normals.)
We therefore develop our own easy necessary and sufficient conditions for realizability, and they can be tested in the WordRAM model.

As a first step, we test that all partial ray angles are well-defined.
Recall that for $i\in \{0,\ldots,n-1\}$ the backward partial ray angle $\rho_i^-$ is the angle between
$\mathbf{c}_i$ and $\mathbf{n}_i$ in the range $[0,90^\circ)$; if the edge normal and corner vector
do not form such an angle then clearly no realization can exist.
(We can test this by computing $\cos \rho_i^-=\langle \mathbf{c}_i,\mathbf{n}_i \rangle$, which must be positive.)
Similarly we require that
$\mathbf{n}_i$ and $\mathbf{c}_{i+1}$ form an angle in the range $[0,90^\circ)$, else the
forward partial ray angle $\rho_i^+$ is not well-defined and no realization exists.

\paragraph{Propagation formula:}     Let us assume that there is a convex realization, say with
corner-distance $d_0$.    The backward partial ray angle satisfies $\cos \rho_0^-=
\langle \mathbf{n}_0,\mathbf{c}_{0} \rangle= h_0/d_0$,
so we know $h_0= d_0 \langle \mathbf{n}_0,\mathbf{c}_{0} \rangle$.
A similar argument with the forward partial ray angle $\rho_0^+$ gives
$d_1= h_0/ \langle \mathbf{n}_0,\mathbf{c}_{1} \rangle = d_0
\langle \mathbf{n}_0,\mathbf{c}_{0} \rangle  / \langle \mathbf{n}_0,\mathbf{c}_{1} \rangle$.
Iterating, we get
$$d_{i+1} = d_0 \prod_{j=0}^i \frac{\langle \mathbf{n}_j,\mathbf{c}_{j} \rangle }{ \langle \mathbf{n}_j,\mathbf{c}_{j+1} \rangle}.$$
\noindent{}Of course $d_{n}$ must equal $d_0$.   So no polygon can exist unless
$$\prod_{i=0}^{n-1} \frac{\langle \mathbf{n}_i,\mathbf{c}_{i} \rangle }{ \langle \mathbf{n}_i,\mathbf{c}_{i+1} \rangle}=1$$
and all the inner products are strictly positive.

Now, we consider the evaluation of the propagation formula (i.e., the computation of the products of the inner products) on the WordRAM.
The multiplication of two $m$-word integers on a WordRAM can be done in linear time with respect to the number of bits of these two numbers~\cite{furer2014}.
Hence, when our word size is $O(\log n)$, we can multiply two $m$-word integers in $f(m) = O(m \log n)$ operations.

To evaluate the propagation formula, we have to multiply $n$ numbers, which may result in numbers that require more computer words than the initial numbers.
Assume that each of the initial $n$ numbers consists of $c$ words in the WordRAM.
Additionally, assume that the numbers (which are not necessarily integers) are encoded in such a way that the multiplication of two numbers requires a constant number of operations on $c$-word integers, including multiplications.

To reduce the number of multiplications with large numbers, we use divide and conquer.
For this, consider a complete binary tree with the initial $n$ numbers at the leaves.
This tree is evaluated bottom-up, where each node is assigned the product of its two children.
With this strategy, all pairwise multiplications are among numbers of roughly equal size, requiring
\begin{equation*}
    T(n) = \sum_{i = 0}^{\log n - 1} \left(\frac{n}{2^{i+1}}\right)  f(c2^i) \in O(n \log^2 n)
\end{equation*}
time altogether.

If equality holds, then fix an arbitrary value of $d_0$.   Using
the propagation formula, this gives us all corner distances.    If we construct
a polygon with these corner distances (and the given corner vectors), then this respects the edge normal due to our
formula, and so is a star-shaped realization.   (It is a convex realization if and only if the polygon angles
are less than $180^\circ$, which we can actually test beforehand by checking whether inner products of consecutive edge normals are positive.)
The only choice we had was $d_0$, but since all other corner distances depend
linearly on $d_0$, a different choice of $d_0$ simply scales the polygon, which is hence unique up to scaling.
We obtain the following result:

\begin{restatable}{theorem}{AlexandrovModified}%
\label{thm:AlexandrovModified}
For any input
sequences of edge normals and corner vectors,
we can test in $O(n \log^2 n)$ time in the \mbox{WordRAM} model whether a realization exists.   If so, then it is unique up to scaling.
\end{restatable}

\subsection{Area and Corner Distance}\label{sec:area-corner_distance}

\begin{theorem}
For input sequences of areas $A_0,\dots,A_{n-1}$ and corner distances $d_0,\dots,d_{n-1}$, we can test in the RealRAM model in $O(n^2\log n)$ time whether it has a star-shaped realization.   In $O(n^3)$ time we can find all such realizations, and test whether any of them is convex.
\end{theorem}
\begin{proof}
Recall that
$\sin\rho_i = 2A_i/d_i d_{i+1}$ for all $i$.
So if $2A_i/d_i d_{i+1}>1$, then the input has no realization. If $2A_i/d_i d_{i+1}=1$, the ray angle $\rho_i$ is $90^\circ$. Otherwise, this gives two possible values for $\rho_i$, one of which is the solution $\alpha_i$ to $\arcsin(2A_i/d_i d_{i+1})$ in $[0,90^\circ)$, and the other is 
$180^\circ-\alpha_i> 90^\circ$.

Assume that we were told the indices $I$ where $\rho_i$ exceeds $90^\circ$;
we have $|I|\leq 3$ since the sum of the ray angles is $360^\circ$. 
We first show that (after some pre-computation) we can then test in constant time whether this choice of $I$ gives a realization, and if so, whether it is convex.   Obviously the ray angles must sum to $360^\circ$, so we need
$\sum_{i\not\in I} \alpha_i +\sum_{i\in I} (180^\circ-\alpha_i) = 360^\circ$.  
If we pre-compute
$\sigma:=\sum_{i} \alpha_i$, then this means testing whether $360^\circ-\sigma$ equals $|I|\cdot 180^\circ - 2\sum_{i\in I} \alpha_i$, which by $|I|\leq 3$ takes constant time.   If this is satisfied, then there is a unique (up to rotation) star-shaped polygon $P_I$ that has these ray angles and corner distances.    To test whether it is convex, observe that the polygon angle at corner $c_i$ is determined by the two ray angles $\rho_{i-1}$ and $\rho_{i}$ since we know all corner distances.   We assume that we have pre-computed the value of $\theta_i$ when these ray angles are $\alpha_{i-1}$ and $\alpha_i$, and flagged all polygon angles above $180^\circ$.   To test whether $P_I$ is convex, we hence only have to re-compute the (up to six) polygon angles adjacent to a ray angle with index in $I$.    For $P_I$ to be convex, these must be less than $180^\circ$, and replace all flagged angles; we can test this in constant time.

So to test whether there is a realization, we only have to find one set $I$ that satisfies the conditions.   With a brute-force approach this can be done in $O(n^3)$ time by checking all possible index sets of size at most 3.  The run-time for star-shaped polygons can be reduced to $O(n^2 \log n)$ as follows.
Store $\alpha_0,\dots,\alpha_{n-1}$ in an array $A$ sorted by value (not index).
Iterate over all subsets $I'$ of at most two indices, and for each of them, find all indices $z$ such that $I'\cup \{z\}$ satisfies the ray angle condition, i.e., 
$2\alpha_z= -360^\circ + \sigma + \left(\left|I'\right|+1\right)\cdot 180^\circ - \sum_{i\in I'} 2\alpha_i$.

The right hand side can be computed in constant time for fixed $I'$.   Therefore, to find a suitable
index $z$ (if any), we do a binary search in $A$, return the range of indices $z$ where $\alpha_z$ has the correct value, and see whether it contains an index not in $I'$; this takes $O(\log n)$ time.

\end{proof}

\subsection{Area and Corner Vector}%
\label{sec:area-corner_vector}

Let us introduce some notation. Given a vector $\mathbf{v}$, we use $\mathbf{v}^{\circlearrowleft}$ to denote the vector obtained by rotating $\mathbf{v}$ by 90$^\circ$ counter-clockwise.

We now study the scenario where we are given the areas $A_0,\dots,A_{n-1}$ as well as the corner vectors $\mathbf{c}_0,\dots,\mathbf{c}_{n-1}$.   Recall  that the ray angle $\rho_i$
is related to the corner vectors; in particular
$\sin\rho_i = \langle \mathbf{c}_{i},\mathbf{c}^{\circlearrowleft}_{i+1} \rangle$ can be computed in the WordRAM model.

\paragraph{Propagation formula:}  The idea to construct a polygon is quite similar to the propagation that we used
in Section~\ref{sec:edge_normal-corner_vector}.
Let us assume that there is a realization, say with
corner-distance $d_0=x$ (with $x$ a variable). 
By the SAS-formula, 
the next corner-dis\-tance is $d_{1}=\tfrac{2A_0}{d_0\sin \rho_0} =\tfrac{2A_0}{\sin\rho_0}\cdot \tfrac{1}{x}$.
Iterating, we get 

\newcommand{\formula}{\noindent\begin{minipage}{\linewidth}
\begin{small}
\begin{align*} 
d_{2i+1} &{=} \frac
{\sin \rho_1 \sin \rho_3 \cdots \sin \rho_{2i-1}} 
{2A_1 \cdot 2A_3 \cdots 2A_{2i-1}}
\frac
{2A_0 \cdot 2A_2 \cdots 2A_{2i}} 
{\sin \rho_0 \sin \rho_2 \cdots \sin \rho_{2i}}
\cdot \frac{1}{x}%
\\[1.1ex]
d_{2i+2} & {=}
\frac{2A_1 \cdot 2A_3 \cdots 2A_{2i+1}}{\sin \rho_1 \sin \rho_3 \cdots \sin \rho_{2i+1}} 
\frac
{\sin \rho_0 \sin \rho_2 \cdots \sin \rho_{2i}} 
{2A_0 \cdot 2A_2 \cdots 2A_{2i}} \cdot x %
\end{align*}%
\end{small}%
\end{minipage}
}%
\formula

\medskip\noindent{}and of course $d_{n}$ must equal $d_0$.   

Consider odd $n$.    From the propagation formula we know that in any realization $d_{n}=s \tfrac{1}{x}$ for some constant $s>0$ that depends only on the input. Since we require $d_{n}=d_0=x$, the only possible realization has $x=\sqrt{s}$,
and for this choice there always exists a solution.
To actually find it, we need to use square-roots, hence the RealRAM model.
If $n$ is even, then we have $d_n=s'x$ for some constant $s'$.  The restriction $d_0=d_{n}=s'd_0$ can be satisfied only if $s'=1$.
If $s'=1$, then for any choice of initial $x$ we get a star-shaped realization.

Using the same approach for multiplication as in \cref{sec:edge_normal-corner_vector} leads to the following
result:

\begin{restatable}{theorem}{AreaCornerVectorsStarshaped}%
\label{thm:AreaCornerVectorsStarshaped}
For any input sequences of corner vectors 
and triangle areas, we can test in $O(n \log^2 n)$ time in the WordRAM model whether a star-shaped realization exists.   
\end{restatable}

We next consider the variant where we want the realization to be convex.

\begin{restatable}{theorem}{AreaCornerVectorsConvex}
For any input sequences of corner vectors 
and triangle areas, we can test in $O(n^2 \log n)$ time in the WordRAM model whether a convex realization exists.   
\end{restatable}

\begin{proof}
Assuming star-shaped realizations exist, write $d_i(x),\theta_i(x)$ etc.\@ for the polygon features obtained if we start propagating with $d_0=x$.
We must be careful in choosing $x$ such that $\theta_i(x)< 180^\circ$ for all $i$.
Note that $\theta_i(x)$ will automatically satisfy this if $\rho_{i-1}+\rho_i\geq 180^\circ$, so we must consider only such indices $i$ where  $\rho_{i-1}+\rho_i<180^\circ$.
Consider the
triangle $T$ spanned by the corners $c_{i-1}(x),c_{i+1}(x)$ and the origin $o$,
see Figure~\ref{fig:area-corner_vector-convex}.
Note that the angle at $c_i(x)$ is smaller/greater than $180^\circ$ if and only if the area of $T$ is  smaller/greater than $A_{i-1}+A_i$.
By the SAS-formula the area of $T$ is
$$ \frac{1}{2} d_{i-1}(x)d_{i+1}(x) \sin\big(\rho_{i-1}{+}\rho_i\big).$$

Note that we can compute $\sin(\rho_{i-1}{+}\rho_i)$ from the inner product of $\mathbf{c}_{i-1}$ and $\mathbf{c}_{i+1}^\circlearrowleft$ in the WordRAM model.

\begin{figure}[ht]
\hspace*{\fill}
\includegraphics[page=1]{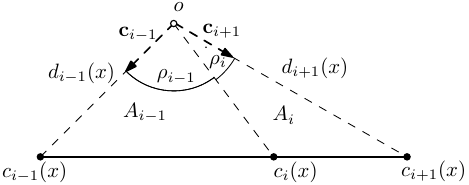}
\hspace*{\fill}
\caption{Configuration with flat angle $\theta_i(x)$.}%
\label{fig:area-corner_vector-convex}
\end{figure}

If $i$ is odd, then from the propagation formula for $\ell=\pm 1$ we have $d_{i+\ell}(x)=s_{i+\ell}x$ 
for some value $s_{i+\ell}$ that only depends on the input and can be computed in the WordRAM model.
Hence, the flat angle happens if
\begin{equation*}    
    x^2=\frac{2(A_{i-1}{+}A_i)}{s_{i+1}s_{i-1}\sin\big(\rho_{i-1}{+}\rho_i\big)}
\end{equation*}
and to have a convex angle $x^2$ needs to be smaller than this.  
If $i$ is even, then similarly $d_{i+\ell}(x)=s_{i+\ell}/x$ for $\ell\in \{\pm 1\}$, the flat angle happens if
\begin{equation*}    
    x^2=\frac{s_{i+1}s_{i-1}\sin\big(\rho_{i-1}{+}\rho_i\big)}{2(A_{i-1}+A_i)}
\end{equation*}
and
we want $x^2$ to be bigger than this.
Over all $i$, this gives us strict upper and lower bounds on $x^2$, and additionally we need $x^2=s$ if $n$ is odd.   
There exists a convex realization if and only if these leave a feasible value for $x^2$.

To calculate the values on a WordRAM, we follow a similar approach as we described for \cref{thm:AreaCornerVectorsStarshaped}.
First, we calculate $s$ as in \cref{thm:AreaCornerVectorsStarshaped} using the divide and conquer approach in $O(n \log^2 n)$ time, with two minor alterations:
here, the factors must be multiplied in order regarding the $s_i$, and we keep the intermediate results.
Notice that the intermediate results relate to each other as a binary tree.
Furthermore, observe that each $s_i$ is the product of $O(\log n)$ intermediate results of word sizes of at most $\frac{n}{2}$, $\frac{n}{4}$, $\frac{n}{8}$, $\ldots$ along a path in that tree.
When multiplying these in ascending order, each $s_i$ can be calculated in $O(n \log n)$ time, with $O(n^2 \log n)$ time altogether.
Afterwards, calculating a product $s_{i+1} s_{i-1}$ requires $O(n \log n)$ time as well, for $O(n^2 \log n)$ total time for all pairs.
This allows calculating all bounds on $x^2$ in $O(n^2 \log n)$ time and finally finding the largest lower bound and the smallest upper bound in $O(n^2)$ time.
\end{proof}

Our two previous algorithms worked under the (standard) assumption that the given areas are labelled in the sense that we know which area is meant for which edge.   We now show that in the unlabelled setting, the problem becomes NP-hard even for a very restricted set of corner vectors.

\begin{restatable}{thm}{CornerVectorsAreaNphard}
    Given numbers $A_0,\dots,A_{n-1}>0$, it is NP-hard to decide whether there exists a star-shaped polygon where all ray angles are equal, and the triangle areas are $A_0,\dots,A_{n-1}$ (not necessarily in this order).
\end{restatable}
\begin{proof}
We give a reduction from the problem \emph{product partition}, where we are given $N$ numbers $a_0,\dots,a_{N-1}$, and we would like to find a subset $I\subseteq \{0,\dots,N{-}1\}$ with
$\prod_{i\in I} a_i = \prod_{i\not\in I} a_i$.
This is known to be NP-hard~\cite{ng_product_2010}.
Define $n=2N$,  
and set $A_i=a_i$ for $i<N$ and $A_i=1$ for $i=N,\dots,n{-}1$.  

Assume that we had a realization with some permutation $\pi$ of these areas, so $A_{\pi(i)}$ is the area of the triangle of $e_i$.   From the formula in Section~\ref{sec:area-corner_vector} and since $n$ is even we know that
\begin{small}
\begin{align*} d_{n} & {=}
\frac
{A_{\pi(1)} \cdot A_{\pi(3)} \cdots A_{\pi(n-1)}}
{\sin \rho_1 \sin \rho_3 \cdots \sin \rho_{n-1}} 
\frac
{\sin \rho_0 \sin \rho_2 \cdots \sin \rho_{n-2}} 
{A_{\pi(0)} \cdot A_{\pi(2)} \cdots A_{\pi(n-2)}} 
\cdot d_0.
\end{align*}
\end{small}

Since all ray angles are equal and $d_{n}=d_0$, we therefore have 
$A_{\pi(1)}\dots A_{\pi(n-1)}=A_{\pi(0)}\dots A_{\pi(n{-}2)}$.
Define $I$ to be the indices of numbers that were areas of even-indexed edges, i.e., 
$I=\{i\in \{0,\dots,N{-}1\}: a_i=A_{\pi(2j)} \text{ for some $j$}\}$.  
Since $A_{N}=\dots=A_{n-1}=1$ do not change the product, we then have $\prod_{i\in I} a_i=\prod_{i\not \in I} a_i$ and hence found the desired subset.

Vice versa, if there is a solution $I$ to the product partition instance, then for $i\in I$ we set $\pi(2i)=i$ (so $A_{\pi(2i)}=a_i$).   For all even indices $2j$ with $j\not\in I$ we use $A_{\pi(2j)}=1$, and for the edges with odd index we use the areas $A_i$ with $i\not\in I$ as
well as the remaining unit areas.
(Note that we had added sufficiently many unit areas for this to be feasible regardless of $|I|$.)   Since $n$ is even, and again using the propagation formula, we see that we get a star-shaped solution when applying the areas in this permutation, and the reduction holds. 
\end{proof}

\subsection{Corner Distance and Edge Distance}{%
\label{sec:corner_distance-edge_distance}

In this section, we study the scenario where we are given sequences of corner distances $d_0,\dots,d_{n-1}$ and edge distances $h_0,\dots,h_{n-1}$.    
From this information we can compute the partial ray angles $\rho_i^+$ and $\rho_i^-$ (at least in RealRAM):
$\cos \rho_i^+=h_i/d_{i+1}$ and $\cos \rho_i^-=h_i/d_i$, and the angles are the solutions in the range $[0,90^\circ)$.   
Unfortunately, this does not necessarily determine the ray angle
(see Figure~\ref{fig:corner_dist_edge_dist_setup}).
We could have $\rho_i=\rho_i^++\rho_i^-$, but
if $d_i>d_{i+1}$, then we could also have $\rho_i=\rho_i^--\rho_i^+$. 

\begin{figure}[ht] \centering
\includegraphics[page=1]{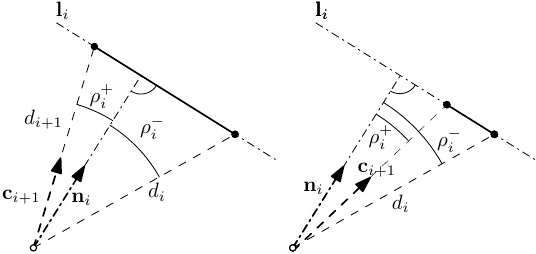}
\caption{Two choices for realizing corner distances and edge distances.}%
\label{fig:corner_dist_edge_dist_setup}
\end{figure}

We exploit this choice (`is the ray angle the sum or the difference of the partial ray angles') to argue that testing whether the input has a star-shaped realization is likely difficult.    More precisely, we provide a reduction from the NP-hard problem 2-partition~\cite{GareyJohnson}, where we are given $N$ positive numbers $a_0,\dots,a_{N-1}$ that sum to $2B$, and we want to know whether there exists a set $I$ of indices such that $\sum_{i\in I} a_i=B$.  
This reduction does not quite prove that the problem is NP-hard because our reduction uses the RealRAM model (see also~\cite{Erickson24} for related discussions), but it certainly suggests that the problem is not easy.

\begin{restatable}{thm}{CornerDistancesEdgeDistancesNPhard}
    There exists a polynomial-time reduction in the RealRAM model from 2-partition to the problem of reconstructing a star-shaped polygon from sequences of corner distances and edge distances.
\end{restatable}
\begin{proof}
Let $a_0,\dots,a_{N-1},B$ be a 2-partition instance.   By padding it with zeroes, we may assume that $N$ is even and we want a partition with $|I|=N/2$.   By adding some large number $m$ to each $a_i$ (and $mN$ to $B$), we may assume that $0<a_i< B/2$ for all $i$.   Finally 
after scaling (and permitting rational numbers for the $a_i$'s), we may assume that $B=30$. 

We will construct sequences of $n=N+2$ corners and edge distances as follows.   Arbitrarily fix $d_0$.
Then, for $1\leq i < N$, define $h_i=d_{i} \cos(3a_i)$, which is positive since $a_i<30$.
In consequence, the backward ray angle $\rho_i^-$ is $3a_i$ in any realization.
Next define $d_{i+1}=h_i/\cos(a_i)=d_i \cos(3a_i)/\cos(a_i)<d_i$, so the forward ray angle $\rho_i^+$ is $a_i$.
Finally close up the input by setting $h_{N}=d_{N}$, $d_{N+1}=d_0$  and 
$$h_{N+1}=d_0\cos\Big(90^\circ-\tfrac{1}{2}\arccos(d_{N}/d_0)\Big).$$
Figure~\ref{fig:corner_dist_edge_dist} shows the construction for a small example.

\begin{figure}[h!]\centering%
\includegraphics[page=3,scale=1.0]{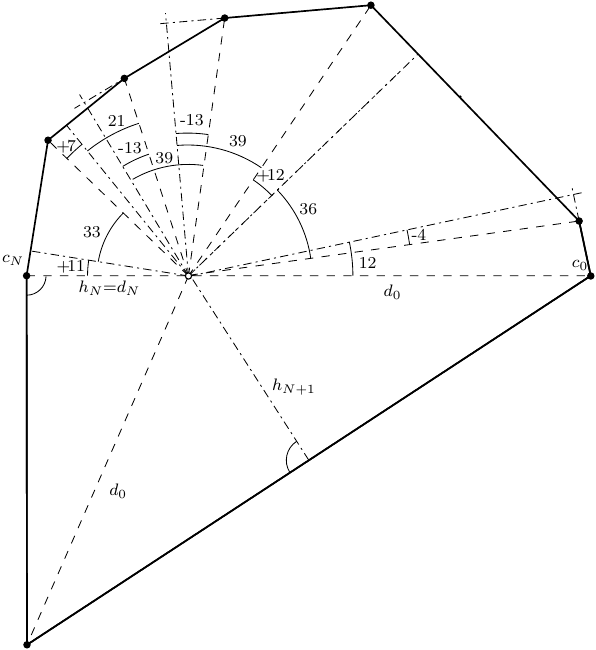}
\caption{Construction of edge distances and corner distances for encoding the 2-partition instance $\{4,12,13,13,7,11\}$.   Each forward ray angle indicates with $\pm$ whether it adds or subtracts to the ray angle.
}%
\label{fig:corner_dist_edge_dist}
\end{figure}

Assume that we have a realization of this input.
Then, for $i<N$, the partial ray angles are $\rho_i^-= 3a_i$ and $\rho_i^+=a_i$.    Also, by $d_{i+1}< d_{i}$, the ray angle $\rho_i$ is either $4a_i$ or $2a_i$.   Let $I$ be the indices of all those ray angles $\rho_i$ that are $2a_i$; then the ray angle between $\mathbf{c}_0$ and $\mathbf{c}_{N}$ is $\sum_{i\in I} 2a_i + \sum_{i\in \{0,\dots,N-1\}\setminus I} 4a_i = 8B-\sum_{i\in I} 2a_i$. 

Now look at the last two triangles $T_{N}$ and $T_{N+1}$.   We set the edge distance $h_{N}$ equal to the corner distance $d_{N}$, which means that $T_{N}$ must have a right angle at $c_{N}$.  With this there is no choice for its ray angle, which must be $\rho_{N}=\arccos(d_{N}/d_0)$. 
Triangle $T_{N+1}$ is isoceles by construction, which means that there is no choice for its ray angle, which is
$$\rho_{N+1}=2\arccos(h_{N+1}/d_0)=180^\circ-\rho_{N}$$
by our choice of $h_{N+1}$.
So $\rho_{N}+\rho_{N+1}=180^\circ$ and
the remaining ray angles sum to $180^\circ=6B$. Since they also sum to $8B-\sum_{i\in I} 2a_i$ we hence have $\sum_{i\in I} a_i = B$.   

So from any realization we can extract a solution to the 2-partition problem. Similarly, one shows how to use a 2-partition solution to construct a star-shaped polygon and the reduction is complete.
\end{proof}

\subsection{Edge Length and Corner Vector}%
\label{sec:edge_length-corner_vector}

Consider the scenario where we are given the edge lengths $\ell_1,\dots,\ell_n$ and the corner vectors $\mathbf{c}_1,\dots,\mathbf{c}_n$.    Here a propagation approach is difficult, because even if we knew (or guessed) the location of corner $c_i$, there are often two possible locations for $c_{i+1}$  where the circle of radius $\ell_i$ centered at $c_i$ intersects the corner ray.  
See Figure~\ref{fig:edge_length_corner_vector}.

\begin{figure}[ht]\centering
\includegraphics[page=1]{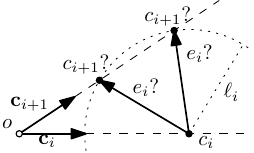}
\caption{Two possibilities for the next corner, even if we know the previous one.}%
\label{fig:edge_length_corner_vector}
\end{figure}

We have not been able to convert this insight into an NP-hardness in the plane, but if we change the model to consider realizations on the cylinder, then we can argue NP-hardness.   Specifically, assume that the origin is at infinity, and the corner rays hence become parallel lines (up to rotation say they are vertical).   A \emph{star-shaped realization on the cylinder} is then the same as a polygonal chain with corners on the corner rays and edge lengths $\{\ell_1,\dots,\ell_n\}$ such that $c_0$ and $c_n$ have the same $y$-coordinate.    The following theorem 
hence says that testing the existence of a star-shaped realization on the cylinder is NP-hard, at least if we allow flat angles.

\begin{restatable}{thm}{CornerVectorsEdgeLengthsNPhard}
    Given sequences of numbers $\ell_1,\dots,\ell_n>0$ and parallel lines $\mathbf{c}_1,\dots,\mathbf{c}_n$ on the flat cylinder, it is NP-hard to decide whether there exists a closed polygonal chain on the flat cylinder with edge lengths $\ell_1,\dots,\ell_n$ and corners on $\mathbf{c}_1,\dots,\mathbf{c}_n$. 
\end{restatable}

\begin{proof}
    We reduce again from the 2-partition problem, so assume that we are given positive numbers $a_0,\dots,a_{n-1},B$ with $\sum_{i=0}^{n-1} a_i=2B$.    Define the vertical lines such that the distance from $\mathbf{c}_i$ to $\mathbf{c}_{i+1}$ is $4a_i$ for $i=0,\dots,n{-}2$, and the flat cylinder has width $8B$ (so the distance from $\mathbf{c}_{n-1}$ to $\mathbf{c}_0$ is $4a_{n-1}$ when ``wrapping around'').
   See also Figure~\ref{fig:edge_length_corner_vector_cylinder}.
Set $\ell_i=5a_i$ for $i\in \{0,\ldots,n-1\}$.

\begin{figure}[ht]\centering
\includegraphics[scale=0.7,page=2]{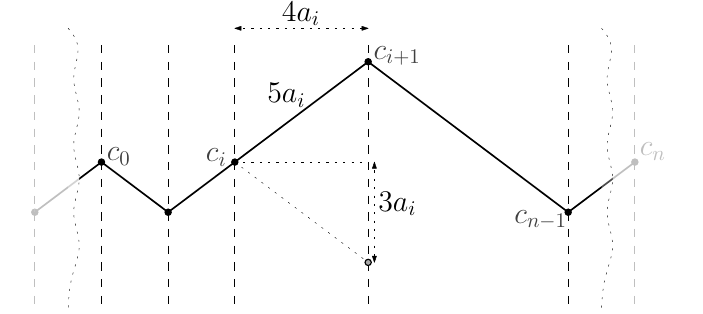}
\caption{NP-hardness construction.  We show the flat cylinder; the standing cylinder is obtained by identifying the curved dotted lines.}%
\label{fig:edge_length_corner_vector_cylinder}
\end{figure}

    Assume we had a realization on the cylinder; up to translation we may assume that $c_0$ has $y$-coordinate $0$.   
    Once we know $c_i$, edge $e_i$ can have either increasing
    or decreasing slope; either way it forms a right triangle with the horizontal through $c_i$ and the (vertical) corner ray for $c_{i+1}$.   The distance between the corner rays is $4a_i$, and the vertical distance between $c_i$ and $c_{i+1}$ hence $\sqrt{\ell_i^2-(4a_i)^2}=3a_i$.   Thus, $e_i$ either adds or subtracts $3a_i$ to the $y$-coordinate of the corner.  Define $I$ to be all those indices where $e_i$ adds $3a_i$ to the $y$-coordinate.   Since $c_n$ must have $y$-coordinate 0 (to match up with $c_0$), we must have $\sum_{i\in I} a_i - \sum_{i\not\in I} a_i=0$ or
    $\sum_{i\in I} a_i=B$ as desired.   

    So from any realization on the cylinder we can extract a solution to the 2-partition problem.   Similarly one shows how to use a 2-partition solution to construct a realization on the cylinder and the reduction is complete.
\end{proof}

\subsection{Area and Edge Normal (the $L_0$-Minkowski Problem)}%
\label{sec:area-edge_normal}

Reconstructing a polygon from triangle areas and edge normals is the $L_0$-Minkowski problem discussed in the introduction,
for which existence of solutions has mostly been solved by Stancu \cite{Stancu2002}: If all edge normals are independent (i.e., we do not have a pair of parallel edges), then there always is a solution.   (She also gives some conditions under which parallel edges can be allowed.) However,  an algorithm to actually reconstruct such polygons remains elusive.

Assume we are given the sequences of edge normals $\mathbf{n}_0,\dots,\mathbf{n}_{n-1}$ as well as areas $A_0,\dots,A_{n-1}$.     We state here a propagation formula that starts with two variables, and can then compute the coordinates of all corners in the WordRAM model:

\begin{itemize}
    \item   Fix the point of corner $c_0$ (its coordinates $(x,y)$ are our two variables).
    \item With this we know the edge distance $h_0=\langle c_0,\mathbf{n}_0\rangle$ since it is the projection of $c_0$ onto the edge ray.
    \item With this, we know the edge-length $\ell_0=2A_0/h_0$.
    \item Since we have the edge-normal, we have all information to compute the next corner:
    $c_1=c_0+\ell_0\mathbf{n}_0^\circlearrowleft$.
\end{itemize}

Repeating the process gives us the coordinates of all corners, dependent on the choice $(x,y)$ of the coordinates of $c_0$.
The formula can all be packed into one as follows:
$$ c_{i+1} = c_{i} + \frac{2A_i}{\langle c_{i},\mathbf{n}_i\rangle} \mathbf{n}_i^\circlearrowleft$$
which clearly can be computed in the WordRAM model.
Ideally, we would now evaluate the condition that the computed corner $c_{n}$ equals the corner $(x,y)$ that we started with, and derive from it which choices of $x$ and $y$ (if any) lead to a realization.   Unfortunately the dependency of $c_{n}$ on $x$ and $y$ is very complex 
and we have not been able to turn this idea into an algorithm to test realizability.

\subsection{Corner Vector and Edge Distance, or Edge Normal and Corner Distance}%
\label{sec:corner_vector-edge_distance}%
\label{sec:edge_normal-corner_distance}%
\label{sec:edge_distance-corner_vector}%
\label{sec:corner_distance-edge_normal}

For convex realizations, the scenario ``corner vectors and edge distance'' is equivalent to ``edge normal and corner distance'', due to polarity.   We will therefore only study the former one.   We have very few results here.   
Given the edge distances, we want the lines that support edges to be tangent to tangent-circles, but even if we know the location of a corner we have a choice of tangent-line for the next edge.
Figure~\ref{fig:onlyEdgeDistancesConvexTwoChoices} shows an example where this can happen even if we want a convex polygon angle.   Specifying the corner vectors does not help with making this choice, but restricts the solution space enough that the general construction of Section~\ref{sec:edge_distance-edge_distance} will usually not work.   So we suspect that the reconstruction problem is NP-hard in this scenario, but this remains open. 

\begin{figure}[ht] \centering
\includegraphics[page=3,trim=0 160 150 0,clip]{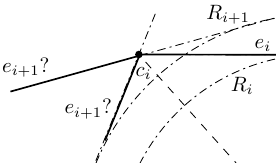}
\caption{Two choices to realize edge distances even for convex polygons.}%
\label{fig:onlyEdgeDistancesConvexTwoChoices}
\end{figure}

\subsection{Two of Area, Edge Length and Edge Distance}%
\label{sec:area-edge_length}%
\label{sec:edge_length_area}%
\label{sec:edge_distance-edge_length}%
\label{sec:edge_length-edge_distance}%
\label{sec:area-edge_distance}%
\label{sec:edge_distance-area}

This section covers three scenarios that are 
all equivalent, since any two of these features can be used to reconstruct the third due to $A_i=\tfrac{1}{2}\ell_i h_i$.
So we can focus on the scenario ``edge distance and area''.
Unfortunately for much of the same reason as in the previous subsection we have very few results here.  As in Figure~\ref{fig:onlyEdgeDistancesConvexTwoChoices} respecting the edge distances alone leaves us with a choice between which tangent-line to use, respecting the area does not help in making this choice, but means that the construction of Section~\ref{sec:edge_distance-edge_distance} usually will not work.   Again we suspect hardness, but leave this for future study.

}

\section{Further Results and Outlook}

Many open problems remain. Some were listed above, but there are also many different scenarios where one studies other features and/or does not specify features in their entirety but assume that only some of their numbers are given.   Finally, we always assumed that the polygon is (at least) star-shaped from the origin; there are many more questions arising when we drop this condition.

\bibliographystyle{plainurl}
\bibliography{references}

\end{document}